\documentclass[12pt]{amsart}
\usepackage{amsfonts}
\usepackage{amssymb}
\usepackage{sgame}
\usepackage{accents}
\usepackage{hyperref}
\usepackage{geometry}                
\usepackage[parfill]{parskip}    
\usepackage{graphicx}
\usepackage{tikz}
\usepackage{natbib}
\usepackage{epstopdf}
\usepackage{graphicx}
\usepackage{bbm}
\usepackage{amscd}
\usepackage{amsmath}
\usepackage{amsthm}
\usepackage{amsfonts}
\usepackage{amssymb}
\usepackage{accents}
\usepackage{indentfirst}
\usepackage{endnotes}
\theoremstyle{definition}

\theoremstyle{remark}
\newtheorem{example}{Example}
\theoremstyle{plain}

\theoremstyle{plain}
\newtheorem{lemma}{Lemma}
\theoremstyle{plain}
\newtheorem{result}{Result}
\theoremstyle{plain}
\newtheorem{definition}{Definition}
\theoremstyle{definition}
\newtheorem{proposition}{Proposition}
\theoremstyle{remark}
\newtheorem{assumption}{Assumption}

\allowdisplaybreaks

\begin{document}
\title{Common Identification and Common Learning}

\author{Martin W. Cripps}
\thanks{Department of Economics, University College London, m.cripps@ucl.ac.uk. My thanks are due to Johannes H\"{o}rner who posed this question and to Duarte Goncalves for his comments and advice. }

\date{June, 2024}

\begin{abstract}
\cite{CrippsElyMailathSamuelson08}  showed that if there are finitely many states, and the signals are i.i.d and finite, then individual learning is sufficient for common learning.
In this note we describe what is commonly learned when  this sufficient condition does not hold. 
\end{abstract}
\maketitle
\section{Introduction}

\cite{CrippsElyMailathSamuelson08} (henceforth CEMS) showed that if agents observe sequences
of iid private signals that ensure they learn an underlying state, then the state will also become approximate common knowledge.
Thus, individual learning is a sufficient condition for common learning.
Here we  provide a characterization of what is commonly learned in settings where individual learning does not occur.
That is, we consider settings where the agents' private signals are not rich enough to ensure that all states are learned. 
As each agent accumulates signals they are only able to learn a partition of the state space: their ``identification partition''. 
Proposition \ref{clll}  shows that the join of the agents' identification partitions determines what is commonly learned.

A motivation for studying common learning is to understand agents' ability to coordinate when they privately learn about the appropriate action. For example, \cite{FrickII23} study how rates of common learning affect equilibrium outcomes and welfare in coordination games. 
However, a maintained assumption in the literature on common learning is that it is \emph{possible} for all agents to learn all of the states.
This seems unlikely to be true in many important macro-based coordination games.
It would, therefore, be useful to have a result that details what is commonly learned when not all agents
can learn all states.
This is the purpose of this note.

In Section 2 we describe the information structures we study and what restrictions we place on them. 
We also define the notion of common identification and give an illustrative example. 
The section ends with a restatement of the definition of common learning for these more general information structures. 
Section 3 contains the result characterising what is commonly learned. The section begins with a description of the events that will be used to establish common learning. Then some well-known results on individual learning are given in Lemma \ref{fr}.
After this, the example is revisited to explain our result on common learning and finally the result is proven.

\section{The Model and Definitions}
There are $L$ agents named $\ell\in\{1,2,\dots,L\}:=[L]$ who each privately observe a sequence of state-dependent signals.
Time is discrete and the time periods are denoted $t=0,1,2,\dots$ \ .
Before period $t=0$, nature selects a state $\theta$ from a finite set $\Theta$ according to a full-support common prior  $p\in\Delta(\Theta)$.
Then, in the periods $t=0,1,2,\dots$, each agent privately observes a signal.
Agent $\ell$ has a finite set of signals $x_\ell\in X_\ell$ and a signal profile is denoted $x=(x_1,\dots,x_L)\in X:=X_1\times\dots\times X_L$.
Conditional on $\theta$, the signals  are generated by independent sampling from the joint distribution
$\pi^\theta\equiv (\pi^\theta(x))_{x\in X} \in\Delta(X)$.
We will denote the marginal over agent $\ell$\/'s signals in state $\theta$ as $\phi^\theta_\ell \equiv(\phi^\theta_\ell(x_\ell))_{x_\ell\in X_\ell}$, 
that is,
$$
\phi_\ell^\theta(x_\ell):=\sum_{x_{-\ell}\in X_{-\ell}}\pi^\theta(x_\ell,x_{-\ell}),
\qquad
\phi_\ell^\theta\in\Delta(X_\ell).
$$
It may not be possible for  agent $\ell$ to learn a particular state from repeated observation of their signals. 
We do not assume that the states generate marginal distributions $\{ \phi_\ell^\theta : \theta\in\Theta\}$ for agent $\ell$ that are distinct.
To describe what the agents are able to learn, we partition
$\Theta$ into the sets of states that agent $\ell$  can identify.

\begin{definition}
For each agent $\ell\in[L]$, their \emph{identification partition}, $\mathcal{Q}_\ell=(Q_\ell^1,\dots,Q_\ell^{n_\ell})$,   is a partition of $\Theta$, such that   $\theta,\theta'\in Q_\ell\in\mathcal{Q}_\ell$ if and only if $\phi_\ell^\theta=\phi_\ell^{\theta'}$.
\end{definition}

 \begin{example}[An Information Structure with Unidentified States]
 \label{eg1}
Consider the following information structure with two agents $\ell\in\{1,2\}$.
There are four states $\Theta=\{\theta_1,\theta_2,\theta_3,\theta_4\}$ and in these states each agent repeatedly observes one of two signals $x_1,x_2\in\{0,1\}$.
The joint distributions of the signals in these four states, $(\pi^{\theta_k})_{k=1}^4$, are given in the tables below.
\begin{figure}[ht]
\begin{center}
\begin{game}{2}{2}
  $\theta_1$& $x_2=0$ & $x_2=1$ \\
$x_1=0$   & $\frac{3}{8}$ &$\frac{1}{8}$  \\
 $x_1=1$   & $\frac{1}{8}$ & $\frac{3}{8}$ 
 \end{game}
 \qquad
 \begin{game}{2}{2}
  $\theta_2$& $x_2=0$ & $x_2=1$ \\
$x_1=0$   & $\frac{5}{12}$ &$\frac{1}{12}$  \\
 $x_1=1$   & $\frac{1}{4}$ & $\frac{1}{4}$ 
 \end{game}
 \\
 \begin{game}{2}{2}
  $\theta_3$& $x_2=0$ & $x_2=1$ \\
$x_1=0$   & $\frac{5}{12}$ &$\frac{1}{4}$  \\
 $x_1=1$   & $\frac{1}{12}$ & $\frac{1}{4}$ 
 \end{game}
 \qquad
 \begin{game}{2}{2}
  $\theta_4$& $x_2=0$ & $x_2=1$ \\
$x_1=0$   & $\frac{2}{5}$ &$\frac{1}{5}$  \\
 $x_1=1$   & $\frac{1}{5}$ & $\frac{1}{5}$ 
 \end{game}
\end{center}
\caption{Four Signal Distributions}
\label{f1}
\end{figure}

\noindent
In these four distributions the marginal over agent 1's signals, $x_1$, in the states $\theta_1$ and $\theta_2$  is $(\frac{1}{2},\frac{1}{2})$,
but in all other states agent 1's marginals are distinct.
By observing enough of her signals, agent 1 can learn the partition $\mathcal{Q}_1=\{ \{\theta_1,\theta_2\},\{\theta_3\},\{\theta_4\}\}$.
This is agent 1's identification partition in this information structure.
Her (conditional) posterior over the set $\{\theta_1,\theta_2\}$ will never vary from her prior.
For agent 2 the signal distributions in the states $\theta_1$ and $\theta_3$ have the marginal $(\frac{1}{2},\frac{1}{2})$, while all other states have distinct marginals.
Thus agent 2 has   the identification partition $\mathcal{Q}_2=\{ \{\theta_1,\theta_3\},\{\theta_2\},\{\theta_4\}\}$.
\end{example}

We will assume (Assumption \ref{fss}) that for every $\theta$ the marginal distribution of any pair of agents' signals has full support.
This assumption is not made in CEMS.
Its purpose here is make the proof of common learning robust to uncertainty in the state that agents cannot detect.
Here agents will not know the state but  only a subset of states.
Thus their inferences about other agents' signals need to be robust to this lack of certainty about the state.
This full-support assumption is a simple way to ensure that this lack of certainty does not prevent common learning of a coarser class of events.
(We will use $x_{-\ell \ell'}$ to denote the signal profile $x$ with the $\ell^{\rm th}$ and $\ell'^{\rm th}$ entry deleted.)

\begin{assumption}
\label{fss}
For all $\theta\in\Theta$, all $\ell,\ell'\in[L]$ and all $(x_\ell,x_{\ell'})\in X_\ell\times X_{\ell'}$, 
$$
\sum_{x_{-\ell \ell'}\in X_{-\ell \ell'}} \pi^\theta(x_\ell,x_{\ell'},x_{-\ell \ell'})>0.
$$
\end{assumption}
This assumption excludes the possibility that  agents' signals are perfectly correlated.
It is slightly more general than the 
 ``fully private'' case by \cite{FrickII23} which requires $\pi^\theta$ to have a full support.

The model of states and signal sequences above  defines a probability space on
the set of states of the world $\Omega:=\Theta\times X^\infty$. The prior $p$, the signal distributions $(\pi^\theta)_{\theta\in\Theta}$ and the i.i.d. assumption  determine a probability measure $\mathbbm{P}$ on $\Omega$.
Associated with $\mathbbm{P}$ is an expectation operator $E(.)$.
We also use $\mathbbm{P}^\theta$ and
$E^\theta(.)$ to denote their $\theta$-conditional versions.
It will often be convenient to abbreviate the event $\{\theta\}\times X^\infty\subset\Omega$ to 
$\{\theta\}$ or just $\theta$.


\subsection{Common Identification and Common Learning}

This section provides definitions of common learning and common identification. 
We begin by defining common identification and then explain how it works in the above example.
Then some additional notation is defined and the section ends with an extension of CEMS's  definition of common learning to sets of states.

\begin{definition}[Common Identification]
\label{cii}
The sets of states that are \emph{commonly identified}  are given by the partition $\mathcal{QI}:=\bigvee_{\ell\in L} \mathcal{Q}_\ell$,
where $\mathcal{Q}_\ell$, $\ell\in[L]$, are the agents' identification partitions.%
\footnote{The join of two partitions $A\vee B$ is the finest partition that has components that are unions of elements of $A$ and of $B$.}
\end{definition}

\setcounter{example}{0}
 \begin{example}[Continued]
In the information structure of Figure  \ref{f1}, agent 1 can identify the partition $\mathcal{Q}_1=\{ \{\theta_1,\theta_2\},\{\theta_3\},\{\theta_4\}\}$ and agent 2 can identify the partition $\mathcal{Q}_2=\{ \{\theta_1,\theta_3\},\{\theta_2\},\{\theta_4\}\}$.
The join of the two partitions is $\mathcal{QI}:=\mathcal{Q}_1 \vee \mathcal{Q}_2=\{ \{\theta_1,\theta_2,\theta_3\},\{\theta_4\}\}$.
The partition $\mathcal{QI}$ contains the events that the agents can \emph{commonly identify} in this setting.
In Proposition \ref{clll} we show that the events in $\mathcal{QI}$ are  commonly learned.
\end{example}

It is now necessary to provide the notation to describe the agents'  learning.
 At the start of period $t$ there is a past history $h_{\ell t}$ of signals for player $\ell$, where $h_{\ell t}=(x_{\ell0},\dots,x_{\ell t-1})\in X_\ell^t$.
The filtration induced by agent $\ell$\/'s histories is denoted $\mathcal{H}_\ell:=(\mathcal{H}_{\ell t})_{t=0}^\infty$.
Agent $\ell$\/'s posterior beliefs about the set $Q\subset\Theta$ of states  is written in the usual way 
$$
E(\mathbbm{1}_{\{\theta\in Q \}} \mid \mathcal{H}_{\ell t})(\omega)\equiv \mathbbm{P}(Q \mid h_{\ell t}(\omega))
\equiv \mathbbm{P}(Q\mid h_{\ell t}).
$$
For any event $F\in\Omega$ we need to define some other events that describe the agents' knowledge about $F$. First  $B^q_{\ell t}(F)$ is the event
 that agent $\ell$ attaches at least probability $q$ to the event $F$ at time $t$.
 Second, $B^q_{t}(F)$  is the event that all agents attach at least probability $q$ to $F$ at time $t$:
$$
B^q_{\ell t}(F):=\{\omega\in\Omega:E(\mathbbm{1}_F\mid \mathcal{H}_{\ell t})\geq q\},
\qquad
B^q_{t}(F):=\bigcap_{\ell\in[L]} B^q_{\ell t}(F).
$$
It is well known that if an individual gets enough signals, then they will learn those states that they can identify. 
That is, individuals' beliefs  eventually will attach high probability to only those states
that could have generated the signals they observe. More formally, for all $\ell\in[L]$, any $q<1$, and any $Q_\ell\in\mathcal{Q}_\ell$
\begin{equation}
\lim_{t\rightarrow\infty} \mathbbm{P}^\theta\left( B^q_{\ell t}( Q_\ell )\right)=1, \qquad \forall \theta\in Q_\ell.
\label{ll}
\end{equation}
A proof of this claim follows from Lemma \ref{fr}  and, for completeness, is given in the Appendix.

We now define common learning of sets of states. After this we report an important sufficient condition for common learning.
The event $F$ is  said to be \emph{common $q$-belief at $t$} on the set of states
\begin{equation}
C^q_{t}(F):=\bigcap_{n\geq1} [B^q_{t}]^n(F).
\label{cl}
\end{equation}
Common learning of a set of states $Q\subset\Theta$ requires that for any $\theta\in Q$ the event $Q$ is eventually common-$q$-belief for any $q$.
That is, for all $q\in(0,1)$  and all $\theta\in Q$
\begin{equation}
\lim_{t\rightarrow\infty}\mathbbm{P}^\theta\left(  C^q_{t}\left(Q\right)\right)=1.
\label{fg}
\end{equation}
The characterization of common $q$-belief of \cite{MondererSamet89} below can be used to write an equivalent but more useful definition of common learning.
The event $F$ is said to be \emph{$q$-evident at $t$} if $F\subset B^q_{t}(F)$.
That is, if the event $F$ occurred all then agents attach at least probability $q$ to it. 
The result below shows that verifying $q$-evidence of an event is a way of establishing common $q$-belief.

\begin{result}[Monderer and Samet]
$F'$ is common $q$-belief at $\omega\in\Omega$ at time $t$ if and only if there exists an event $F$ such that $F$ is $q$-evident at time $t$ and
$\omega\in F\subset B^q_t(F')$.
\end{result}

As common $q$-belief (for any $q$) is necessary and sufficient for common learning,
 CEMS show that the previous definition for common learning (\ref{cl}) can be rewritten in terms of $q$-evidence.

\begin{definition}[Common Learning]
Individuals commonly learn the event $Q\subset\Theta$,
 if for all $q\in(0,1)$ there exists a  time $T$ 
 and a sequence of events $(F_{t})_{t=0}^\infty$ such that
 for all $t\geq T$ and for all $\theta\in Q$: (1)
$F_t\subset B^q_t(Q)$; (2)
$\mathbbm{P}^\theta\left(  F_t \right)>q$;
(3)
$F_t$ is $q$-evident at $t$.
\label{defcl}
\end{definition}

Notice that if $Q$ is a set with only one element this is equivalent to the previous definitions of common learning.


\section{Common Learning}

In this section we will show that the sets in the partition of commonly identified sets, $\mathcal{QI}$, are commonly learned.
We being by describing some features of individual learning and then move on to proving the main result.

\subsection{Individual Learning}
Here we define  events on signals that will later be used to define $q$-evident events.
These are different from the events used in previous proofs of common learning.
The $q$-evident sets of signals are defined in a way that depends on $Q_\ell$ but they are independent of the particular  $\theta\in Q_\ell$.
This is because  an agent does not know which state in $Q_\ell$ is actually generating their signals.
It  contrasts with CEMS and \cite{FrickII23}, where the $q$-evident sets are sometimes chosen in a way that depends on the details of the joint 
distributions of signals  $\pi^\theta$.
In those constructions, signal profiles that have zero probability affect the fine details of the $q$-evident event for each state.
Assumption \ref{fss} allows us avoid this issue.
Also, the total variation norm, $\Vert.\Vert_{TV}$,
 will be used.%
 \footnote{ $\Vert\phi_\ell\Vert_{TV}:=\frac{1}{2}\sum_{x_\ell}|\phi_\ell(x_\ell)|$.}
 This norm permits the use of the Dobrushin Ergodicity Coefficient, which has useful contraction properties when applied to Markov-like operators.
CEMS use the $\sup$ norm and \cite{FrickII23} who use the relative entropy distance to build a contraction.

The $q$-evident events necessary to establish common learning are constructed so that all agents observe signals with empirical measures close to their marginals in a state.
For each agent $\ell$ and each signal $x_\ell\in X_\ell$, let $\hat x_{\ell}(h_{\ell t})\in \mathbbm{Z}_+$ count the number of times the signal $x_\ell$ occurs in the history $h_{\ell t}$. 
Then, $\hat\phi_{\ell t}\in\Delta(X_\ell)$ is the period-$t$ empirical measure of agent $\ell$\/'s signals,
where $\hat\phi_{\ell t}(x_\ell):=t^{-1}\hat x_{\ell}(h_{\ell t})$.
If all agents' empirical measures at time $t$ are within $\varepsilon$ of
their marginal distribution of signals in  state $\theta$, then $I_{t\varepsilon}(\theta)$ has occurred.
That is,
$$
I_{t\varepsilon}(\theta):=\{\omega: \Vert\hat\phi_{\ell t}-\phi^\theta_{\ell}\Vert_{TV}\leq \varepsilon, \forall \ell\in[L]\}.
$$
The earlier proofs of common learning showed that neighborhoods like $I_{t\varepsilon}(\theta)$ are common 
$q$-belief.
In our setting it is  possible that no agent can  identify the state $\theta$ or this neighborhood,
so we  consider  unions of these neighborhoods over the sets of states.
The event $I_{t\varepsilon}(Q)$ is defined so that there is some $\theta\in Q$ such that all agents have observed signals close to the marginals in the state $\theta$.

\begin{definition}
For any $\varepsilon>0$ and $Q\subset\Theta$ let $I_{t\varepsilon}(Q):=\cup_{\theta\in Q}I_{t\varepsilon}(\theta)$.
\end{definition}

The following lemma now records some well-know properties of the events $I_{t\varepsilon}(Q)$ when $Q$ is chosen to lie in the partition of commonly identified sets:  $(I_{t\varepsilon}(Q))_{Q\in\mathcal{QI}}$ .
It shows (in parts (1) and (2)) that if $\theta\in Q\in\mathcal{QI}$, then the agents' empirical distributions are in the sets $I_{t\varepsilon}(\theta)$ and $I_{t\varepsilon}(Q)$ with arbitrarily high probability for $t$ sufficiently large.
This follows from the fact that signals are iid and the strong laws of large numbers.
It also shows (part (3)) that if  $Q_\ell\in\mathcal{Q}^\ell$ is in agent $\ell$\/'s identification partition and $Q_\ell\ni\theta$, 
then a lower bound on the agent's posterior  on $Q_\ell$ converges to one on the event $I_{t\varepsilon}(\theta)$ as $t$ grows.
This uses the fact that on $I_{t\varepsilon}(\theta)$ the log likelihoods grow linearly in $t$.
Hence (in part 4) $I_{t\varepsilon}(\theta)$ is a subset of the event where $Q_\ell$ is $q$-believed by agent $\ell$ for a suitably chosen $q$.
The proof of the Lemma is given in the appendix.
It uses  standard results on learning that can be found many places in the literature.

\begin{lemma}
\label{fr}
There exists $\bar\varepsilon>0$ such that for all $\varepsilon\in(0,\bar\varepsilon)$ and all $Q\in\mathcal{QI}$
\begin{enumerate}
\item
For all $\theta$, $\mathbbm{P}^\theta(I_{t\varepsilon}(\theta))\rightarrow1$ as $t\rightarrow\infty$.
\item
For all $\theta\in Q$, $\mathbbm{P}^\theta(I_{t\varepsilon}(Q))\rightarrow1$ as $t\rightarrow\infty$.
\item
There exists $b>0$ such that for any $\theta\in \bigcap_{\ell=1}^L Q_\ell\subset Q$, $\ell\in[L]$, and  $\omega\in I_{t\varepsilon}(\theta)$,
 $\mathbbm{P}(Q_\ell|h_{\ell t}(\omega)) \geq 1-\frac{e^{-tb}}{p^\theta}$.
\item
There exists $b>0$ such that for any $\theta\in Q_\ell\in\mathcal{Q}_\ell$, any $\ell$, and any $t$,
$$
I_{t\varepsilon}(\theta)\subset B_{\ell t}^q(Q_\ell),
\qquad
\makebox{if} 
\qquad q=1-\frac{e^{-tb}}{p^\theta}.
$$
\end{enumerate}
\end{lemma}

\setcounter{example}{0}
\begin{example}[Continued]
There are two commonly identified sets of states: $\{\theta_1,\theta_2,\theta_3\}$ and $\{\theta_4\}$.
We will, therefore, want to consider the events $I_{t\varepsilon}(\{\theta_1,\theta_2,\theta_3\})$ and $I_{t\varepsilon}(\{\theta_4\})$.
The last of these, $I_{t\varepsilon}(\{\theta_4\})$, occurs if the  empirical measure of agent 1's signals and the empirical measure of agent 2's signals are within $\varepsilon$ of $(3/5,2/5)$. 
(This is the blue region in Figure \ref{f2}.)
The lemma says that 
if $\theta_4$ is the true state then this will happen with high probability for $t$ large.
Furthermore, on this event
both agents will  attach high probability the $\theta_4$.
The event $I_{t\varepsilon}(\{\theta_1,\theta_2,\theta_3\})$ occurs if:  the empirical measure of
agent 1's signals is within $\varepsilon$ of $(1/2,1/2)$ and the empirical measure of agent 2's signals 
is  within  $\varepsilon$ of either $(1/2,1/2)$ or $(2/3,1/3)$;
or if  the empirical measure of
agent 1's signals is within $\varepsilon$ of $(2/3,1/3)$ and the empirical measure of agent 2's signals 
is  within  $\varepsilon$ of $(1/2,1/2)$. (In pink in the figure.)
These are illustrated in the figure below where the axes describe the probabilities of the zero signal.
When $I_{t\varepsilon}(\{\theta_1,\theta_2,\theta_3\})$ occurs and agent 1 observes frequencies close to $\frac{1}{2}$ they do not know whether agent 2 has seen frequencies close to $\frac{1}{2}$ or $\frac{2}{3}$.
Similarly, if agent 2 observes frequencies close to $\frac{1}{2}$ they do not know whether agent 1 has seen frequencies close to $\frac{1}{2}$ or $\frac{2}{3}$.

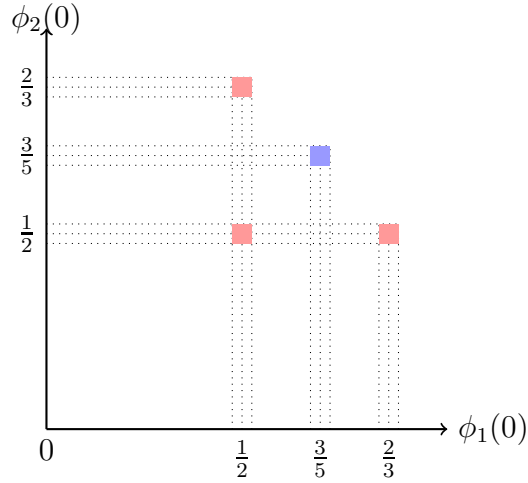
\begin{figure}
\begin{center}
\begin{tikzpicture}[scale=1.3]
\draw [thick,->] (0,0) -- (4.1,0);
\node [right] at (4.1,0) {$\phi_1(0)$};
\node [below] at (0,0) {$0$};
\node at (0,4.2) {$\phi_2(0)$};
\draw [thick,->] (0,0) -- (0,4.1);
\node [below] at (3.5,0) {$\frac{2}{3}$};
\node [below] at (2.8,0) {$\frac{3}{5}$};
\node [left] at (0,3.5) {$\frac{2}{3}$};
\node [left] at (0,2.8) {$\frac{3}{5}$};
\node [below] at (2,0) {$\frac{1}{2}$};
\node [left] at (0,2) {$\frac{1}{2}$};
\draw [dotted] (3.5,0) -- (3.5,2);
\draw [dotted] (3.6,0) -- (3.6,2.1);
\draw [dotted] (3.4,0) -- (3.4,2.1);
\draw [dotted] (2.7,0) -- (2.7,2.9);
\draw [dotted] (2.8,0) -- (2.8,2.8);
\draw [dotted] (2.9,0) -- (2.9,2.9);
\draw [dotted] (2.1,0) -- (2.1,3.6);
\draw [dotted] (2,0) -- (2,3.5);
\draw [dotted] (1.9,0) -- (1.9,3.6);
\draw [dotted] (0,3.6) -- (2.1,3.6);
\draw [dotted] (0,3.5) -- (2,3.5);
\draw [dotted] (0,3.4) -- (2.1,3.4);
\draw [dotted] (0,2.1) -- (3.6,2.1);
\draw [dotted] (0,2) -- (3.5,2);
\draw [dotted] (0,1.9) -- (3.6,1.9);
\draw [dotted] (0,2.8) -- (2.8,2.8);
\draw [dotted] (0,2.7) -- (2.9,2.7);
\draw [dotted] (0,2.9) -- (2.9,2.9);
\fill[blue!40!white] (2.7,2.7) rectangle (2.9,2.9);
\fill[red!40!white] (1.9,1.9) rectangle (2.1,2.1);
\fill[red!40!white] (1.9,3.4) rectangle (2.1,3.6);
\fill[red!40!white] (3.4,1.9) rectangle (3.6,2.1);
\end{tikzpicture}
\end{center}
\caption{$I_{t\varepsilon}(\theta_4)$ in Blue and $I_{t\varepsilon}(\theta_1,\theta_2,\theta_2)$ in Pink}
\label{f2}
\end{figure}

\end{example}


\subsection{Common Learning}

What follows is the main/only result of this paper. It shows that the sets in the partition $\mathcal{QI}$ are commonly learned.
The proof simplifies the original construction of CEMS which allows us to deal with the case that there are several states that agents cannot detect.
The way to proof works is first explained in the example below.

\setcounter{example}{0}
\begin{example}[Continued]
Consider the event $I_{t\varepsilon}(\{\theta_1,\theta_2,\theta_3\})$ illustrated in Figure \ref{f2}.
The proof shows that this is common $q$-belief for $t$ sufficiently large.
If this event has occurred, then agent 1 will either have seen signals that are close to  $(\frac{1}{2},\frac{1}{2})$  or  close $(\frac{2}{3},\frac{1}{3})$.
Suppose that at time $t$ agent 1 has observed signals with the empirical probability $(\frac{1}{2}+\nu,\frac{1}{2}-\nu)$.
That is, agent 1 has observed signals distance $\nu$ from the true distribution of signals in states $\theta_1$ and $\theta_2$. 
If $\nu$ is small and $t$ is large, agent 1 will only attach high probability to these two states.
Agent 1's expected frequencies agent 2's signals will depend on which state ($\theta_1$ or $\theta_2$) occurred.
Conditional on $\theta_1$ she will do the calculation below to form her expectation of the empirical measure of agent 2's signals:
$$
({\textstyle \frac{1}{2}}+\nu,{\textstyle \frac{1}{2}}-\nu)\left[
\begin{array}{cc} 
\frac{3}{4} & \frac{1}{4}\\
\frac{1}{4} & \frac{3}{4}
\end{array}\right]
=
({\textstyle \frac{1+\nu}{2}},{\textstyle \frac{1-\nu}{2}}).
$$
Conditional on $\theta_2$ she will do the different calculation:
$$
({\textstyle \frac{1}{2}}+\nu,{\textstyle \frac{1}{2}}-\nu)\left[
\begin{array}{cc} 
\frac{5}{6} & \frac{1}{6}\\
\frac{1}{2} & \frac{1}{2}
\end{array}\right]
=
({\textstyle \frac{2+\nu}{3}},{\textstyle \frac{1-\nu}{3}}),
$$
Whichever of the states $\theta_1$ or $\theta_2$ occurred, her expectations of agent 2's empirical measure are even closer (in the variation norm) to agent 2's marginals than agent 1's signals are.%
\footnote{For this to be true in general requires Assumption \ref{fss}.}
If agent 1 observes signals close to the marginal in the partition element $\{\theta_1,\theta_2\}$, they do not know which of these two states has occurred, but if it was $\theta_1$ (respectively $\theta_2$) then agent 2's signals are in $I_{t\varepsilon}(\{\theta_1\})$ (respectively $I_{t\varepsilon}(\{\theta_2\})$) with high probability.
As $I_{t\varepsilon}(\{\theta_1,\theta_2\})\subset I_{t\varepsilon}(\{\theta_1,\theta_2,\theta_3\})$ they also attach high probability to this event.
Repeating this procedure for all agents and all elements of the identification partitions that lie in $\{\theta_1,\theta_2,\theta_3\}$ we can show that this event
is $q$-believed for $t$ large.

This  illustrates the principal difference between this result and ones where all states are identified.  It is necessary that expectations are closer in the variation norm for \emph{all} of the states in the partition element $\{\theta_1,\theta_2\}$, thus the fine details of the joint distributions of the signals in one of these states cannot be used to define $q$-evident events. 
\end{example}

\begin{proposition}
The event $Q\subset\Theta$ is commonly learned if and only if $Q\in \mathcal{QI}$.
\label{clll}
\end{proposition}

\begin{proof}
We begin by proving sufficiency.
Let $Q^*\in \mathcal{QI}$ and $q\in(0,1)$ be given. 
Choose $\varepsilon\in(0,\bar\varepsilon)$ as in Lemma \ref{fr} and $\beta\in(0,1)$ sufficiently small so that $q^\beta(1-L(1-q^\beta))>q$.
We will show that the sequence of events $(I_{t\varepsilon}(Q^*))_{t=0}^\infty$ has the properties required of the sequence $(F_t)_{t=0}^\infty$ in Definition \ref{defcl}.
Hence we will show that $Q^*$ is commonly learned.

By Part 2 of Lemma \ref{fr}, $\mathbbm{P}^\theta(I_{t\varepsilon}(Q^*))\rightarrow 1$ for any $\theta\in Q^*$.
Choose $T'$ so that $\mathbbm{P}^\theta(I_{t\varepsilon}(Q^*))>q^\beta$ for all $t\geq T'$ and all $\theta\in Q^*$.
We, then, will choose
$$
T\geq T'\wedge \max_{\theta}\frac{\log(1-q^\beta)p_\theta}{-b} .
$$

\emph{(1) To show  $I_{t\varepsilon}(Q^*)\subset B^q_t(Q^*)$ for all $t\geq T$:}
Note that the above lower bound on $T$ ensures that  $q^\beta\leq 1-e^{-tb}/p^\theta$ for all $t\geq T$ and  $\theta$.
Pick $\theta\in \bigcap_{\ell=1}^L Q_\ell \subset Q^*$, then by Part 3 of 
Lemma \ref{fr}, $I_{t\varepsilon}(\theta)\subset B_{\ell t}^{q^\beta}(Q_\ell)$ for all $\ell\in[L]$ and $t\geq T$.
As $Q_\ell\subset Q^*$ this implies $I_{t\varepsilon}(\theta)\subset B_{\ell t}^{q^\beta}(Q^*)$ for all $\ell$ and all $t\geq T$.
Therefore,  $I_{t\varepsilon}(\theta)\subset B_{t}^{q^\beta}(Q^*)$  for all $t\geq T$.
This inclusion holds for all $\theta\in Q^*$, so we can take a union over the events on the left to get 
$I_{t\varepsilon}(Q^*)\subset B_{t}^{q^\beta}(Q^*)$
for all $t\geq T$.
The claim $I_{t\varepsilon}(Q^*)\subset B^q_t(Q^*)$ follows as $\beta<1$.

\emph{(2) To show that for all $\theta\in Q^*$ and all $t\geq T$, $\mathbbm{P}^\theta(I_{t\varepsilon}(Q^*))>q$:} This follows from the choice of $T\geq T'$ and $\beta<1$.

(3) \emph{$I_{t\varepsilon}(Q^*)$ is $q$-evident for all $t>T$:}
Pick $\omega\in I_{t\varepsilon}(Q^*)$, then by definition there must be some $\theta\in \bigcap_{\ell=1}^L Q^*_\ell \subset Q^*$ such that  $\omega\in I_{t\varepsilon}(\theta^*)$. Hence, at $\omega$ the  empirical measure of each agent's signals is close to their marginals under $\theta^*$: $\Vert\hat\phi_{\ell t}-\phi^{\theta^*}_\ell\Vert_{TV}<\varepsilon$ for all $\ell$.

Now we define a family of matrices for each state $\theta\in\Theta$.
The matrix $M^\theta_{\ell \ell'}$ is $|X_\ell|\times |X_{\ell'}|$ and has rows that describe agent $\ell$'s conditional beliefs about agent $\ell'$'s signals in the state $\theta$.

The entries in the matrix are $\pi^\theta(x_{\ell},x_{\ell'})$, the marginal distribution on agents $\ell$ and $\ell'$'s signals, normalized by 
$\phi_\ell^\theta(x_\ell)$
the marginal on agent $\ell$\/'s signals. Therefore, these non-square matrices have rows that sum to unity like Markov matrices.
$$
M^\theta_{\ell \ell'}:=\left(\frac{\pi^\theta(x_{\ell},x_{\ell'})}{\phi_\ell^\theta(x_\ell)}\right)_{x_\ell\in X_\ell, x_{\ell'}\in X_{\ell'}},
\makebox{where}
\ \ \ 
\pi^\theta(x_{\ell},x_{\ell'}):=\sum_{x_{-\ell \ell'}\in X_{-\ell \ell'}} \pi^\theta(x_\ell,x_{\ell'},x_{-\ell \ell'}).
$$
(Recall that will use $x_{-\ell \ell'}$ to denote the signal profile $x\in X$ with the $\ell^{\rm th}$ and $\ell'^{\rm th}$ entries omitted.)
If we treat $\phi^\theta_\ell$ and $\phi^\theta_{\ell'}$ as column vectors then elementary arithmetic shows that
\begin{equation}
(\phi^\theta_\ell)^TM^\theta_{\ell \ell'}=(\phi^\theta_{\ell'})^T,
\qquad
\forall \ell\not=k.
\label{marg}
\end{equation}

We now describe the \emph{Doeblin contraction coefficient} or \emph{Dobrushin's ergodic coefficient} of the matrix $M^\theta_{\ell \ell'}$.
Let $m^\theta_{\ell \ell'}(x_\ell)\in\Delta(X_{\ell'})$ be the $x_\ell^{\rm th}$ row of  $M^\theta_{\ell \ell'}$.
By Assumption \ref{fss}, these conditional probability distributions have full support, therefore, for all $\theta$ and all $\ell\not=\ell'$
\begin{align*}
\Vert m^\theta_{\ell \ell'}(x_\ell)-m^\theta_{\ell \ell'}(x'_\ell)\Vert_{TV}&<1,\qquad \forall x_\ell,x'_\ell\in X_\ell.
\end{align*}
(The variation distance will equal one only if the measures have disjoint support.)
We will define $\lambda^\theta_{\ell \ell'}<1$ to be the maximum variation distance between the rows of $M^\theta_{\ell \ell'}$, that is,
$$
\lambda^\theta_{\ell \ell'}:=\max_{x_\ell,x'_\ell \in X_\ell}\Vert m^\theta_{\ell \ell'}(x_\ell)-m^\theta_{\ell \ell'}(x'_\ell)\Vert_{TV}.
$$ 
This maximum is also strictly less than one by finiteness.
The value $\lambda^\theta_{\ell \ell'}$ is known as Doeblin contraction coefficient or Dobrushin's Ergodic Coefficient of the matrix $M^\theta_{\ell \ell'}$.
Let $\lambda<1$ be the largest value of $\lambda^\theta_{\ell \ell'}$, that is $\lambda:=\max_{\ell,\ell',\theta}\lambda^\theta_{\ell \ell'}$.
By Theorem 7.3 \cite{Bremaud99} p.237 we have that for any $\ell\not=\ell'$, any $\phi_\ell,\phi'_\ell\in\Delta(X_\ell)$ and any $\theta$
\begin{equation}
\Vert (\phi_\ell)^TM^\theta_{\ell \ell'} - (\phi'_\ell)^TM^\theta_{\ell \ell'} \Vert_{TV}\leq \lambda^\theta_{\ell \ell'} \Vert \phi_\ell - \phi'_\ell\Vert_{TV}\leq \lambda \Vert \phi_\ell - \phi'_\ell\Vert_{TV}.
\label{dob}
\end{equation}
Thus the matrices $M^\theta_{\ell \ell'}$ describe linear operators from $\Delta(X_\ell)$ to $\Delta(X_{\ell'})$ that are contractions in the variation norm and $\lambda$ is an upper bound this contraction.

By CEMS Lemma 3 (p.923), $(\hat\phi_{\ell t})^TM^\theta_{\ell \ell'}$ is agent $\ell$\/'s best prediction of agent $\ell'$'s empirical measure.
To be precise, it is says that for any $q<1$ there exists a $T_3$ such that for all $\theta$ all $\ell'\not=\ell$ and all $t\geq T_3$ 
\begin{align}
\mathbbm{P}^\theta\left( \left\Vert (\hat\phi_{\ell t})^TM^\theta_{\ell \ell'} -(\hat\phi_{\ell' t})^T\right\Vert_{TV}\leq (1-\lambda)\varepsilon \mid h_{\ell t}\right)>q^\beta,
\label{cems}
\end{align}
(Note: In CEMS the norm is the twice the total variation norm that is used here.)

Now we will combine the conditions (\ref{dob}) and (\ref{cems}) to show that when 
$\omega\in I_{t\varepsilon}(\theta^*)$ agent $\ell$ attaches high probability to the event  $I_{t\varepsilon}(Q^*_\ell)$ having occurred.  
Consider agent $\ell$ at $\omega\in I_{t\varepsilon}(\theta^*)\subset I_{t\varepsilon}(Q^*)$.
Conditional any $\theta$ we derive an upper bound on how far agent $\ell'$\/'s empirical measure is from its true marginal in terms on agent $\ell$\/'s marginal.
First notice we have the following upper bound on this:
\begin{align*}
\Vert \hat \phi_{\ell' t}-\phi^\theta_{\ell'}\Vert_{TV}
&=\Vert (\hat \phi_{\ell' t})^T-(\phi^\theta_{\ell})^TM^\theta_{\ell \ell'}\Vert_{TV}
\\
&\leq
\Vert (\hat \phi_{\ell' t})^T-(\hat\phi_{\ell t})^TM^\theta_{\ell \ell'}\Vert_{TV}
+
\Vert (\hat\phi_{\ell t})^TM^\theta_{\ell \ell'}-(\phi^\theta_{\ell})^TM^\theta_{\ell \ell'}\Vert_{TV}
\\
&\leq
\Vert (\hat \phi_{\ell' t})^T-(\hat\phi_{\ell t})^TM^\theta_{\ell \ell'}\Vert_{TV}
+
\lambda \Vert \hat\phi_{\ell t}-\phi^\theta_{\ell}\Vert_{TV}
\end{align*}
(The first equality substitutes from (\ref{marg}),
then there is an application of the triangle inequality and 
the final upper bound uses (\ref{dob}).)

Suppose that $\omega\in I_{t\varepsilon}(\theta^*)$.
As $\phi^{\theta^*}_{\ell}=\phi^\theta_{\ell}$ for all
$\theta\in Q^*_\ell$, we know that
 $\Vert\hat\phi_{\ell t}-\phi^{\theta}_\ell\Vert_{TV}\leq\varepsilon$ for all $\theta\in Q^*_\ell$.
 Thus  if $\omega\in I_{t\varepsilon}(\theta^*)$
$$
\Vert \hat \phi_{\ell't}-\phi^\theta_{\ell'}\Vert_{TV}
\leq
\Vert (\hat \phi_{\ell' t})^T-(\hat\phi_{\ell t})^TM^\theta_{\ell \ell'}\Vert_{TV}
+
\lambda \varepsilon,
\qquad 
\forall\theta\in Q^*_\ell .
$$
By (\ref{cems}) the first term on the right is less than $\varepsilon(1-\lambda)$ with at least probability $q^\beta$ in any state $\theta$.
Thus the right is less than $\varepsilon$  with at least probability $q^\beta$ in all states $\theta\in Q^*_\ell$.
That is, for all $t>\max\{T_3,T\}$, all $\omega\in I_{t\varepsilon}(\theta^*)$, all $\ell'\not=\ell$ and all $\theta\in Q^*_\ell$
\begin{equation}
\mathbbm{P}^\theta\left(\left\Vert \hat \phi_{\ell' t}-\phi^\theta_{\ell'}\right\Vert_{TV}
\leq \varepsilon\mid h_{\ell t} \right) >q^\beta
\label{df}
\end{equation}
 If agent $\ell$ observes a history consistent with $\omega\in I_{t\varepsilon}(\theta^*)$,
 the event $ I_{t\varepsilon}(\theta^*)$ occurs if $\Vert \hat \phi_{\ell' t}-\phi^{\theta^*}_{\ell'}\Vert_{TV}
\leq \varepsilon$ for all agents $\ell'\not=\ell$, because agent $\ell$ knows their signals satisfy this condition.
From the lower bound (\ref{df}), the probability that all $L-1$ of these events occurs is at least $1-L(1-q^\beta)$.
Hence we have that if $\omega\in I_{t\varepsilon}(\theta^*)$ for all $\theta\in Q^*_\ell$
\begin{equation}
\mathbbm{P}^\theta( I_{t\varepsilon}(\theta) \mid h_{\ell t} )
=
\mathbbm{P}^\theta(\Vert \hat \phi_{\ell' t}-\phi^\theta_{\ell'}\Vert_{TV}
\leq \varepsilon, \forall \ell'\not=\ell \mid h_{\ell t} ) >1-L(1-q^\beta)
\label{fin}
\end{equation}

Finally we can establish the $q$-evidence.
Recall that $Q^*\in\mathcal{QI}$ and that $Q^*_\ell\subset Q^*$.
We begin by doing the following calculation which uses the fact that conditional on $\theta\in Q^*_\ell$ agent $\ell$\/'s beliefs are equal to their priors.
\begin{align*}
\mathbbm{P}(I_{t\varepsilon}(Q^*)\mid h_{\ell t})
&\geq
\mathbbm{P}(I_{t\varepsilon}(Q^*_\ell)\mid h_{\ell t})
\\
&\geq
\mathbbm{P} (Q^*_\ell\mid h_{\ell t})\mathbbm{P}(I_{t\varepsilon}(Q^*_\ell)\mid h_{\ell t},\theta\in Q^*_\ell)
\\
&=
\mathbbm{P} (Q^*_\ell\mid h_{\ell t}) \frac{\sum_{\theta\in Q^*_\ell} p_\theta \mathbbm{P}^\theta(I_{t\varepsilon}(Q^*_\ell)\mid h_{\ell t})}{\sum_{\theta'\in Q^*_\ell} p_{\theta'}}
\\
&\geq
\mathbbm{P} (Q^*_\ell\mid h_{\ell t}) 
\frac{\sum_{\theta\in Q^*_\ell} p_\theta \mathbbm{P}^\theta(I_{t\varepsilon}(\theta)\mid h_{\ell t})}{\sum_{\theta'\in Q^*_\ell} p_{\theta'}}
\end{align*}

By (\ref{fin}) if $\omega\in I_{t\varepsilon}(\theta^*)$ then we have a lower bound on all of the terms in the final summation above. Hence for all $t>\max\{T_3,T\}$ and all $\omega\in I_{t\varepsilon}(\theta^*)$
we have that 
$$
\mathbbm{P}(I_{t\varepsilon}(Q^*)\mid h_{\ell t})\geq \mathbbm{P} (Q^*_\ell\mid h_{\ell t}) (1-L(1-q^\beta)).
$$
By our choice of $T$ we also have that  $\mathbbm{P} (Q^*_\ell\mid h_{\ell t}) \geq q^\beta$.
Thus we have that for $t>\max\{T_3,T\}$ and all $\omega\in I_{t\varepsilon}(\theta^*)$
$$
\mathbbm{P}(I_{t\varepsilon}(Q^*)\mid h_{\ell t})\geq q^\beta (1-L(1-q^\beta))>q.
$$
(The final inequality follows from our choice of $\beta$.)
Using the language of beliefs this shows that $I_{t\varepsilon}(\theta^*)\subset B^q_\ell(I_{t\varepsilon}(Q^*))$.
As this holds for all agents $\ell$ and for any $\theta^*\in Q^*$ this implies $I_{t\varepsilon}(Q^*) \subset B^q(I_{t\varepsilon}(Q^*))$ and so 
$I_{t\varepsilon}(Q^*)$ is $q$-evident for all $t>\max\{T,T_3\}$ and we have completed the proof of Step 3.

Finally, we need to show that $Q\in \mathcal{QI}$ is necessary for common learning.
Consider a set $\hat Q\not\in\mathcal{QI}$ we will show that this  cannot be commonly learned.
As $\hat Q\not\in\mathcal{QI}$, the set $\hat Q$ is not generated by the information partition of some agent $\ell$.
Hence there exists $\tilde\theta,\hat\theta\in Q_\ell\in\mathcal{Q}_\ell$ 
so that $\hat\theta\in \hat Q$ but $\tilde\theta\not\in\hat Q$.
By (\ref{ll}) we know that for any $r<1$
$$
\lim_{t\rightarrow\infty} \mathbbm{P}^{\hat\theta}(B^r_{\ell t}(Q_\ell))=1.
$$
Let $p^{\tilde\theta}$ be the prior probability of state $\tilde\theta$.
As the relative probabilities that agent $\ell$ attaches to states in $Q_\ell$ do not change as they observe their signals,
 agent $\ell$ must attach at least probability $rp^{\tilde\theta}$ to state $\tilde\theta$
if she attaches probability $r$ to $Q_\ell$, hence for any $r<1$
$$
\lim_{t\rightarrow\infty} \mathbbm{P}^{\hat\theta}(B^{rp^{\tilde\theta}}_{\ell t}(\{\tilde\theta\}))=1.
$$
Choose $\bar q\in(1-rp^{\tilde\theta},1)$.
If you attach probability $rp^{\tilde\theta}$ to an event it is impossible to attach probability $\bar q$ to its complement.
Therefore, as $\tilde\theta\not\in\hat Q$ the above implies
\begin{equation}
\lim_{t\rightarrow\infty} \mathbbm{P}^{\hat\theta}(B^{q}_{\ell t}(\hat Q))=0,
\qquad
\forall q>\bar q.
\label{sd}
\end{equation}
From the definitions of common learning (\ref{cl}) and (\ref{fg}), $\hat Q$ is commonly learned if  for all $q\in(0,1)$ and for all
$\theta\in\hat Q$
$$
\lim_{t\rightarrow\infty}\mathbbm{P}^\theta\left(  C^q_{t}\left(\hat Q\right)\right)=1
\qquad 
\makebox{where} \qquad C^q_{t}(\hat Q):=\bigcap_{n\geq1} [B^q_{t}]^n(\hat Q).
$$
But 
$C^q_{t}\left(\hat Q\right)\subset B^q_{t}(\hat Q)\subset B^q_{\ell t}(\hat Q)$ and $\hat\theta\in\hat Q$
so
$$
\lim_{t\rightarrow\infty}\mathbbm{P}^{\hat\theta}\left(  C^q_{t}\left(\hat Q\right)\right)\leq 
\limsup_{t\rightarrow\infty}\mathbbm{P}^{\hat\theta}\left(  B^q_{\ell t}(\hat Q) \right)
$$
However, by (\ref{sd}) the right hand side is bounded by zero if $q>\bar q$.
This means common learning is does not hold.
\end{proof}

\section{Conclusion}
We have generalised the results of CEMS to the case where not all agents learn all the states.
This generalisation is not without cost, however, as it was necessary to assume that signals have full support.

\renewcommand{\baselinestretch}{1}
\normalsize

\bibliographystyle{econometrica}

\bibliography{mwccl}
\addcontentsline{toc}{section}{~~~~References}


\section*{Appendix}

\subsection*{Proof of Lemma \ref{fr}}

Let $\theta\in \bigcap_{\ell=1}^L Q_\ell$ be given, where every $Q_\ell\in\mathcal{Q}_\ell$, and let 
$$
H( \phi_\ell \Vert \phi_\ell'):=\sum_{x_\ell}\phi_\ell(x_\ell)\log\frac{\phi_\ell(x_\ell)}{\phi'_\ell(x_\ell)},
\qquad
\phi_\ell,\phi'_\ell\in\Delta(X_\ell);
$$ 
denote the Kullback-Leibler divergence (assuming $\phi'_\ell(x_\ell)>0$).
Also, $N_\varepsilon(\phi_\ell):=\{\phi'_\ell\in\Delta(X_\ell):\Vert\phi_\ell-\phi'_\ell\Vert_{TV}\leq \varepsilon\}$ will be used to denote an $\varepsilon$-neighbourhood of $\phi_\ell$ in the total variation norm.
If $\theta\in \bigcap_{\ell=1}^L Q_\ell$, then in states $\theta\not\in Q_\ell$ agent $\ell$ observes signals sampled from distinct  distributions from those generated by signals in $Q_\ell$.
So, if $\phi_\ell=\phi_\ell^\theta$ then $H(\phi_\ell \Vert \phi_\ell^\theta)-H(\phi_\ell \Vert \phi_\ell^{\theta'})<0$  for all $\theta'\not\in Q_\ell$.
As $H$ is a continuous function of $\phi_\ell$, the finite set of strict inequalities 
$$
H(\phi_\ell \Vert \phi_\ell^\theta)-H(\phi_\ell \Vert \phi_\ell^{\theta'})<0, 
\qquad
\theta'\not\in Q_\ell;
$$ 
continues to hold for $\phi_\ell$ close to $\phi_\ell^\theta$.
As there are a finite number of states and agents we can choose a neighborhood size so these inequalities hold for all states and agents.
That is,
there exists an $\bar\varepsilon>0$ such that for all $\theta\in\Theta$, all $\ell$, and all $\varepsilon\in(0,\bar\varepsilon)$
$$
H(\phi_\ell\Vert \phi_\ell^\theta)-H(\phi_\ell\Vert\phi_\ell^{\theta'})< 0,
\qquad \forall \phi_\ell\in N_{\varepsilon}(\phi_\ell^\theta), \theta\in Q_\ell, \theta'\not\in Q_\ell.
$$
Given this finite collection of strict inequalities there exists $b>0$, dependent on $\varepsilon$, so that for all $\theta$ and all $\ell$
\begin{equation}
H(\phi_\ell\Vert \phi_\ell^\theta)-H(\phi_\ell\Vert\phi_\ell^{\theta'})<-b,
\qquad \forall \phi_\ell\in N_{\varepsilon}(\phi_\ell^\theta), \theta\in Q_\ell, \theta'\not\in Q_\ell.
\label{chen}
\end{equation}

\paragraph{Part (1):}
Sanov's Theorem (see for example \cite{CoverThomas91} p.292) gives bounds on the probability that $\hat\phi_{\ell t}\in N_{\varepsilon}(\phi_\ell^\theta)$ conditional on $\theta$ being the true state.
\begin{equation}
1-(t+1)^{-|X_\ell|} e^{-t \alpha_{\ell\theta}(\varepsilon)}
\geq 
\mathbbm{P}^\theta\left( \hat\phi_{\ell t}\in N_{\varepsilon}(\phi_\ell^\theta)\right)\geq 1-(t+1)^{|X_\ell |} e^{-t \alpha_{\ell\theta}(\varepsilon)}.
\label{san}
\end{equation}
where $\alpha_{\ell\theta}(\varepsilon):= \inf\{ H( \phi_\ell\Vert\phi_\ell^\theta) :  \phi_\ell \in \Delta(X_\ell)\setminus N_{\varepsilon}(\phi_\ell^\theta) \}$.
The event $I_{t\varepsilon}(\theta)$ is the intersection of the events $\hat\phi_{\ell t}\in N_{\varepsilon}(\phi_\ell^\theta)$ for $\ell\in[L]$.
Thus 
$$
\mathbbm{P}^\theta(I_{t\varepsilon}(\theta))\geq1-\sum_{\ell\in[L]} (t+1)^{|X_\ell |} e^{-t \alpha_{\ell\theta}(\varepsilon)},
$$
and the claim follows.

\paragraph{Part (2):}
This is true as $I_{t\varepsilon}(\theta)\subset I_{t\varepsilon}(Q)$ and part (1) holds.

\paragraph{Part (3):}
Let $\hat\phi_{\ell t}$ be the empirical probabilities of agent $\ell$\/'s signals at time $t$. Then,
\begin{align*}
\log \mathbbm{P}(\theta' |h_{\ell t}) 
&\leq 
\log \frac{\mathbbm{P}(\theta' |h_{\ell t})}{ \mathbbm{P}(\theta |h_{\ell t})}= \log \frac{\mathbbm{P}(\theta' \cap h_{\ell t})}{ \mathbbm{P}(\theta \cap h_{\ell t})}
\\
&=\log\frac{p^{\theta'}}{p^\theta}+t\sum_{x_\ell\in X_\ell }\hat\phi_{\ell t}(x_\ell )\log\frac{\phi_\ell^{\theta'}(x_\ell)}{\phi_\ell^{\theta}(x_\ell)}
\\
&=\log\frac{p^{\theta'}}{p^\theta}+t\left[H(\hat\phi_{\ell t}\Vert \phi_\ell^{\theta})-H(\hat\phi_{\ell t}\Vert \phi_\ell^{\theta'})\right]
\end{align*}
(Here the second equality follows from a substitution for the probabilities of the history $h_{\ell t}$.)
For any $\varepsilon<\bar\varepsilon$ and $\hat\phi_{\ell t}\in N_{\varepsilon}(\phi_\ell^\theta)$ and $\theta'\not\in Q_\ell$ there is an upper bound of $-b$ on the term in square parentheses in (\ref{chen}). 
Hence, for all $\theta'\not\in Q_\ell$ and all $\hat\phi_{\ell t}\in N_{\varepsilon}(\phi_\ell^\theta)$,
$$
\mathbbm{P}(\theta' |h_{\ell t}) 
\leq
\frac{p^{\theta'}}{p^\theta}e^{-tb}.
$$
A summation then implies that for $\hat\phi_{\ell t}\in N_{\varepsilon}(\phi_\ell^\theta)$ 
\begin{equation}
\mathbbm{P}(Q_\ell |h_{\ell t}) = 
1-\sum_{\theta'\not\in Q_\ell}\mathbbm{P}(\theta' |h_{\ell t}) 
\geq
1-
\frac{e^{-tb}}{p^\theta}.
\label{gb}
\end{equation}

\subsection*{Proof of  (\ref{ll})}

Conditional on $\theta\in \bigcap_{\ell=1}^L Q_\ell$ the random variables $\mathbbm{P}(Q_\ell|h_{\ell t})$ are  bounded submartingales
with respect to the filtrations $(\mathcal{H}_{\ell t})$.
By Doob's Theorem they converge almost surely.
It remains to show that this limit is one.
This can be established by combining (\ref{gb}) and (\ref{san}). That is, $\mathbbm{P}(Q_\ell|h_{\ell t})
\geq
1-
\frac{1}{p^\theta}e^{-tb}$ with at least probability $1-(t+1)^{|X_\ell|} e^{-t \alpha_\theta(\varepsilon)}$ conditional on $\theta$.

\end{document}